\documentclass[amsmath, amssymb, floatfix, reprint, onecolumn, notitlepage, 10pt]{revtex4-1}

\usepackage{amsmath}
\usepackage{amssymb}
\usepackage{amsthm}
\usepackage{mathrsfs}
\usepackage{bbm}
\usepackage{tikz}
\usetikzlibrary{decorations.markings,decorations.pathreplacing,calc}
\usetikzlibrary{arrows.meta}
\usepackage{subcaption}
\usepackage{array}

\usepackage[framemethod=tikz]{mdframed}
\usepackage[margin=1in]{geometry}
\usepackage[utf8]{inputenc}
\usepackage{comment}
\usepackage{umoline}

\usepackage{hyperref}
\hypersetup{
     colorlinks=true,           
     linkcolor=purple,          
     citecolor=blue,            
     filecolor=blue,            
     urlcolor=cyan,             
 }

\usepackage[capitalise]{cleveref}

\theoremstyle{plain}
\newtheorem{theorem}{Theorem}
\newtheorem{lemma}[theorem]{Lemma}

\theoremstyle{definition}

\newtheorem{claim}[theorem]{Claim}

\newtheorem{conjecture}[theorem]{Conjecture}

\newcommand{\eps}{\varepsilon}
\newcommand{\id}{\mathbbm{I}}
\newcommand{\cd}{\mathrm{CC}}
\newcommand{\br}[1]{\left(#1\right)}
\newcommand{\Tr}[1]{\mathrm{Tr}\left(#1\right)}
\newcommand{\cR}{\mathcal{R}}
\newcommand{\cG}{\mathcal{G}}
\newcommand{\cH}{\mathcal{H}}
\newcommand{\EPR}{\mathrm{EPR}}

\newcommand {\bra} [1] {\ensuremath{ \left\langle #1 \right| }}
\newcommand {\ket} [1] {\ensuremath{ \left| #1 \right\rangle }}
\newcommand {\ketbratwo} [2] {\ensuremath{ \left| #1 \middle\rangle \middle\langle #2 \right| }}
\newcommand {\ketbra} [1] {\ketbratwo{#1}{#1}}
\newcommand{\braket}[2]{\left\langle#1 | #2 \right\rangle}

\begin{document}

\title{A construction of Combinatorial NLTS}
\author{Anurag Anshu}
\email{anuraganshu@fas.harvard.edu}
\affiliation{School of Engineering and Applied Sciences, Harvard University, Cambridge, MA, USA}
\author{Nikolas P. Breuckmann}
\email{n.breuckmann@ucl.ac.uk}
\affiliation{Department of Computer Science, University College London, WC1E 6BT London, United Kingdom}
\date{\today}

\begin{abstract}
    The NLTS (No Low-Energy Trivial State) conjecture of Freedman and Hastings \cite{FH14} posits that there exist families of Hamiltonians with all low energy states of high complexity (with complexity measured by the quantum circuit depth preparing the state). Here, we prove a weaker version called the combinatorial NLTS, where a quantum circuit lower bound is shown against states that violate a (small) constant fraction of local terms. This generalizes the prior NLETS results \cite{EldarH17, NirkheVY18}. Our construction is obtained by combining tensor networks with expander codes \cite{SipserS96}. The Hamiltonian is the parent Hamiltonian of a perturbed tensor network, inspired by the `uncle Hamiltonian' of~\cite{GSWCP15}. Thus, we deviate from the quantum CSS code Hamiltonians considered in most prior works.
\end{abstract}

\maketitle

\section{Introduction}

The approximation of the ground energy of a local Hamiltonian continues to be a leading goal of quantum complexity theory and quantum many-body physics. While a generic, accurate and efficient approximation method is unlikely, due to the seminal result of Kitaev~\cite{quant-ph/0210077}, physically motivated ansatzes such as tensor networks \cite{Orus2017} and low depth quantum circuits \cite{Peruzzo14, LaRose2019, AGMS21} continue to explore the low energy spectrum of many interesting Hamiltonians.   

A fundamental question on the power of low-depth quantum circuits is the NLTS conjecture \cite{FH14}, which posits the existence of local Hamiltonians with all low energy states having high quantum circuit complexity. This is a necessary consequence of the quantum PCP conjecture \cite{10.1145/2491533.2491549}, under the reasonable assumption that $\mathrm{QMA}\neq \mathrm{NP}$. We refer the reader to existing works \cite{10.1145/2491533.2491549, EldarH17, NirkheVY18, ANirkhe22, BravyiKKT19} for a detailed discussion on the NLTS conjecture and its close connection with quantum error correction, robustness of entanglement and the power of variational quantum circuits. 

To formally define the NLTS conjecture, we introduce a $n$-qubit local Hamiltonian $H$ as a sum of local terms $H=\sum_{i=1}^m h_i$ (each $0\preceq h_i \preceq \id$ is supported on $O(1)$ qubits and each qubit participates in $O(1)$ local terms) with $m=\Theta(n)$. The ground states of $H$ are the eigenstates with eigenvalue $\lambda_{\min}(H)$. An $\eps$-energy state $\psi$ satisfies $\Tr{H\psi}\leq \eps m + \lambda_{\min}(H)$.

\begin{conjecture}[NLTS \cite{FH14}]
There exists a fixed constant $\eps > 0$ and an explicit family of $O(1)$-local Hamiltonians $\{H^{(n)}\}_{n=1}^\infty$, such that for any family of $\eps$-energy states $\{\psi_n\}$,
the circuit complexity $\cd(\psi_n)$ grows faster than any constant.
\end{conjecture}
Here, $\cd(\psi)$ is \textit{quantum circuit depth}, the depth of the smallest quantum circuit that prepares $\psi$. An interesting property of any (potential) NLTS Hamiltonian is that it must live on an expanding interaction graph, ruling out all the finite-dimensional lattice Hamiltonians that have been very well studied in quantum many-body physics. The same holds for (potential) Hamiltonians that may witness the quantum PCP conjecture \cite{10.1145/2491533.2491549}. 

A weaker version of this conjecture is known, called the NLETS theorem. A local Hamiltonian $H$ (as defined above) is frustration-free if $\lambda_{\min}(H)=0$. A state $\psi$ is called \emph{$\eps$-error} if there exists a set $S$ of qubits of size at least $(1-\eps)n$ such that $\psi_S=\phi_S$, where $\phi$ is some ground state of $H$ and the subscript $S$ means that we take a partial trace over the qubits in $[n]\setminus S$. 
\begin{theorem}[\cite{EldarH17, NirkheVY18}]
There exists a fixed constant $\eps > 0$ and an explicit family of $O(1)$-local frustration-free Hamiltonians $\{H^{(n)}\}_{n=1}^\infty$, such that for any family of $\eps$-error states $\{\psi_n\}$ , the circuit complexity $\cd(\psi_n)$ is $\Theta(\log n)$.
\end{theorem}
Note that the $\eps$-energy states include the set of $O(\eps)$-error states, but the reverse direction is not true. The NLETS theorem was first proved by \cite{EldarH17}, by considering the hypergraph product \cite{TillichZ09}  of two Tanner codes on expander graphs \cite{SipserS96}. In the follow-up work, \cite{NirkheVY18} constructed an NLETS Hamiltonian that in fact lived on a one-dimensional lattice. In the recent work \cite{ANirkhe22}, super-constant circuit lower bounds were shown for $o(1)$-energy states (such as $O(\frac{1}{\log n})$-energy) of all quantum code Hamiltonians that have near-linear rank or near-linear distance. Interestingly, such lower bounds are again possible with the two-dimensional, punctured toric code, showing that expansion of the underlying interaction graph is not needed for circuit lower bounds on `almost constant' energy states. 

Both \cite{EldarH17, NirkheVY18} identified the intermediate question of \textit{combinatorial NLTS} (cNLTS), which aims at finding frustration-free Hamiltonians with super-constant circuit lower bounds for states $\psi$ that satisfy at least $1-\eps$ fraction of local terms. The main interest in this question stems from the fact that any (potential) cNLTS Hamiltonian must also live on an expanding interaction hypergraph, hence exhibiting the geometric features of an NLTS Hamiltonian. Here, we provide the first construction of a cNLTS Hamiltonian. 

\begin{theorem}[Main result]
\label{thm:cNLTS}
There exists a fixed constant $\eps > 0$ and an explicit family of $O(1)$-local frustration-free Hamiltonians $\{H^{(n)}\}_{n=1}^\infty$, where $H^{(n)}=\sum_{i=1}^m h_i^{(n)}$ acts on $n$ particles and consists of $m=\Theta(n)$ local terms, such that for any family of states $\{\psi_n\}$ satisfying 
$$\frac{|i: \Tr{h^{(n)}_i\psi_n}>0|}{m}\leq \eps,$$
the circuit complexity $\cd(\psi_n)$ is $\Theta(\log n)$.
\end{theorem}

The set of states that satisfy $1-\eps$ fraction of local terms also include $O(\eps)$-error states. Thus, the above family of Hamiltonians are also NLETS. 

\subsection*{Other related results}
 In \cite{Eldar2019RobustQE}, thermal states of certain quantum codes were shown to have circuit lower bounds. \cite{BravyiKKT19} showed circuit lower bounds for quantum states with a `$Z_2$ symmetry'.  The work \cite{EldarH17} showed that locally testable quantum CSS codes \cite{AharonovE15} of linear distance are NLTS. Such codes are not known to exist, with the best distance thus far being $\sqrt{n}$ \cite{LeverrierLZ22}. However, the dramatic recent progress in quantum codes \cite{HastingsHO21, Breuckmann2020BalancedPQ, panteleevk2020, panteleevk2021} opens up the exciting possibility that such codes may exist.

\vspace{0.1in}

\textit{New results:} The follow-up work \cite{ABN22} supersedes the main result here, as it shows the NLTS property for good quantum codes \cite{LeverrierLZ22, panteleevk2021}. It uses a different Hamiltonian family, but the underlying connection is that it proves a quantum analogue of Theorem \ref{thm:dualdist}. We believe that the Hamiltonian family in this work are also NLTS, when the parameter $\delta$ is set to a constant.

\subsection*{Outline of the construction}

Our starting point is the NLETS theorem shown in \cite{NirkheVY18}. It is based on the observation that the CAT state $\frac{1}{\sqrt{2}}\ket{00\ldots 0}+\frac{1}{\sqrt{2}}\ket{11\ldots 1}$ is close to the unique ground state of Kitaev's clock Hamiltonian. This clock Hamiltonian is obtained from the circuit preparing the CAT state and then padding with identity gates. We observe that yet another Hamiltonian can be constructed, by viewing the CAT state as a Matrix Product State (MPS). The MPS representation of the CAT state is obtained by starting with $\frac{n}{2}$ EPR pairs
$$\br{\ket{00}+\ket{11}}_{1,2}\otimes \br{\ket{00}+\ket{11}}_{3,4}\otimes \ldots \br{\ket{00}+\ket{11}}_{n-1, n}$$
and then projecting qubits $i,i+1$, for even $i$, with the projector $M=\ketbra{00}+\ketbra{11}$. Most MPS are the unique ground states of a parent Hamiltonian (such MPS are called injective). But the CAT state MPS clearly does not have this privilege, since $M$ is not an invertible map. However, inspired by \cite{GSWCP15}, we can perturb $M$ to consider a state obtained by mapping qubits $i,i+1$ with $M+\delta \id$ (for $\delta\approx \frac{1}{\sqrt{n}}$). This is an invertible map, which makes the resulting MPS injective. Using the corresponding parent Hamiltonian, we obtain another construction of the NLETS Hamiltonian.\footnote{One may note that the local terms of the resulting parent Hamiltonian do not have norm $\leq 1$ and may scale as $\frac{1}{\delta}$. But we can fix this by the standard trick of replacing each local term by the projector onto its image. Frustration-free property of the parent Hamiltonian ensures that the ground state is unchanged.}

Since any cNLTS Hamiltonian must be on an expanding interaction graph, an approach to construct the desired Hamiltonian is to write down a tensor network for the CAT state on an expanding graph, perturb the tensors and then take the parent Hamiltonian. Unfortunately, this argument seems not to work, since the tensor network for the CAT state is extremely brittle. If we remove one EPR pair and allow arbitrary inputs to the tensors acting on this EPR pair, we can produce the states $\ket{00\ldots 0}$ or $ \ket{11\ldots 1}$. This brittleness reflects in the nearby parent Hamiltonian and there are product states that violate just one local term. Our second observation is that the CAT state tensor network can be viewed as a \textit{repetition Tanner code} on an expanding graph. Thus, we can generalize the tensor network and look at Tanner codes defined on expander graphs, as proposed by Sipser-Spielman \cite{SipserS96} (Section~\ref{sec:tannertensor}). The tensor network state is now a uniform superposition over all the codewords of this Tanner code.
With (1) linear distance and (2) linear rank (and with a suitable choice of parameters) such a code protects us from two sources of brittleness:
\begin{enumerate}
    \item Removing an $\eps$ fraction of EPR pairs (analogously, local terms of the Hamiltonian) weakens the expansion properties of the underlying graph. Linear distance ensures that the codewords, while no longer far away from each other, are partitioned into distant groups for a small constant $\eps$.
    \item Removing an $\eps$ fraction of EPR pairs (analogously, local terms of the Hamiltonian) can drastically reduce the number of strings appearing in the superposition. Linear rank ensures that the number of strings is large enough, if $\eps$ is a small constant.
\end{enumerate}
See Section \ref{sec:cNLTS} for full details. We note that tensor networks have previously been combined with local (quantum) codes to obtain global properties \cite{Pastawski2015}. 

\subsection*{Local systems and non-isotropic Gau\ss's laws}

Tanner codes can be understood in terms of homology with local systems~\cite{steenrod1943homology}, where differentials take values in the space of local checks \cite{meshulam2018graph}.
A trivial example is the toric code, where the differential at each vertex detects violations of $\mathbb{Z}_2$-flux conservation, or in other words, violations of a local parity-check code (this is of course nothing but the usual simplicial $\mathbb{Z}_2$-homology).
The family of Hamiltonians that we construct can be understood in terms of differentials defined from more complicated local codes.
The Hamiltonians ensure that these differentials are zero for ground states, which means they enforce a non-isotropic Gau\ss 's law that takes the directionality of the incoming fluxes into account.
Together with the expansion of the underlying graph, this leads to the cNLTS property.

\subsection*{Tensor networks and quantum complexity}

Kitaev's clock construction is a powerful method to map quantum computations to the ground states of local Hamiltonians. It turns out that the tensor networks provide a similar mapping. As shown in \cite{SchuchWVC07}, any measurement-based quantum computation can be mapped onto a tensor network. One could thus imagine a form of circuit-to-Hamiltonian mapping different from Kitaev's: perturb the above tensor network and consider its parent Hamiltonian. A standard objection to this approach is that the mapping also works for post-selected quantum circuits, which is far more powerful than QMA. However, this objection is not expected to apply to our case, as injective tensors cannot post-select on events of very small probability. We leave an understanding of the promise gap of this mapping for future work.     

\section{Tanner code}\label{sec:tanner_code}

Consider a regular graph $G=(V,E)$ with degree $d$ and $n=|V|$ vertices. For $S,S'\subset V$, we denote the number of edges between $S$ and $S'$ as $E(S,S')$ (we count an edge $\{u,v\}$ twice if $u,v\in S\cap S'$). 
Let $\lambda=\max\br{|\lambda_2|, |\lambda_n|}$, where $\lambda_2, \lambda_n$ are the second largest and the smallest eigenvalues of the adjacency matrix.  

A Tanner code $T(C,G)\subset \{0,1\}^{|E|}$ is defined using the graph $G$ and a classical linear code $C \subset \{0,1\}^d$ of rank $k_0$ and distance $\Delta_0$. We imagine bits on edges and checks on the vertices. Let the edges be numbered using the integers $\{1,2,\ldots |E|\}$ in some arbitrary manner. Given a string $x\in\{0,1\}^{|E|}$ and a vertex $v$, let $x_v\in \{0,1\}^{d}$ be the restriction of~$x$ to the edges incident to $v$, where the $i$th bit of $x_v$ is the value on the edge with the $i$th smallest number. Formally, 
\begin{equation}
\label{eq:tandef}
T(C,G)=\{x: x_v\in C \quad \forall v\in V\}.    
\end{equation}

We will abbreviate $T(C,G)$ as $T$ for convenience.
Since there are $d-k_0$ independent checks in $C$, the number of independent checks in $T$ is at most $n(d-k_0)$. Thus, the rank $k$ of $T$ is $k\geq \frac{nd}{2}-n(d-k_0)=n\br{k_0-\frac{d}{2}}$. 
\begin{lemma}
\label{lem:tandist}
Suppose $\Delta_0\geq 2\lambda$. The distance of $T$ is lower bounded by $\frac{n\Delta^2_0}{4}=\frac{|E|\Delta^2_0}{2d}$.
\end{lemma}
\begin{proof}
Let $x\in T$ be the non-zero code-word of smallest Hamming weight and let $E_x$ be the edges where $x$ takes value $1$. Let $S$ be the set of all vertices on which at least one edge in~$E_x$ is incident. Since the distance of $C$ is $\Delta_0$, at least $\Delta_0$ edges from any vertex in~$S$ stay within~$S$ (and those edges belong to $E_x$). Thus, $|E(S, S)| \geq |S|\Delta_0$. However, the expander mixing lemma \cite[Lemma 4.15]{Vad12} ensures that 
$$|E(S, S)| \leq \frac{d|S|^2}{n} + \lambda|S|.$$ Thus,
$$|S|\Delta_0 \leq |E(S, S)| \leq \frac{d|S|^2}{n} + \lambda|S|\leq \frac{d|S|^2}{n} + \frac{\Delta_0|S|}{2}\implies |S|\geq \frac{n\Delta_0}{2}.$$
Since $|E_x|\geq \frac{|S|\Delta_0}{2}$ (every vertex in $S$ is associated to at least $\Delta_0$ edges in $E_x$ and every edge in $E_x$ is associated to $2$ vertices in $S$), the lemma concludes.
\end{proof}
   
\section{Injective tensor network from the code $T(C,G)$}
\label{sec:tannertensor}

Let $\ket{\EPR} = \ket{00}+\ket{11}$ be an unnormalized EPR pair. Given $G$, we consider a Hilbert space consisting of $nd$ qubits, with $d$ qubits for each vertex $v\in V$. For a $v\in V$, we identify each qubit with a unique edge incident on $v$ and label the qubit as $v_e$. As a result, given an edge $e=(v,v')$, qubits $v_e, v'_e$ come in pairs (Figure \ref{fig:fullgraph}). We will often denote the joint Hilbert space $\cH_{v_e}\otimes \cH_{v'_e}$ as $\cH_e$ and abbreviate $\ket{0}_{v_e}\ket{0}_{v'_e}$ as $\ket{0}_e$ and $\ket{1}_{v_e}\ket{1}_{v'_e}$ as $\ket{1}_e$. Thus, $\ket{0}_{v_e}\ket{0}_{v'_e}+\ket{1}_{v_e}\ket{1}_{v'_e}$ will be referred to as $\ket{\EPR}_e$. Define the unnormalized state
$$\ket{\Theta_0}:= \bigotimes_{e\in E}\ket{\EPR}_{e}.$$
For each vertex $v$, define the projector that only accepts the codewords of the local code at~$v$: $$P_v:= \sum_{c\in C}\ketbra{c_1}_{v_{e_1}}\otimes \ketbra{c_2}_{v_{e_2}}\otimes \dotsb \otimes \ketbra{c_d}_{v_{e_d}},$$ where $c_i$ is the $i$th bit of $c$ and $e_i$ is the $i$th edge incident on $v$ (in the ascending numbering specified on the edges). The tensor network state is obtained by projecting the $d$ qubits on each vertex using these projectors:
$$\ket{\Phi} := \frac{1}{\sqrt{2^k}}\bigotimes_v P_v \ket{\Theta_0} = \frac{1}{\sqrt{2^k}}\sum_{x\in T} \bigotimes_{e=(v,v')}\ket{x_e}_{e}.$$
Note that the normalization follows since there are $2^k$ codewords in $T$. We can think of the string $x$ as an \textit{edge assignment} and $\ket{\Phi}$ as a uniform superposition over edge assignments from $T$. 
\begin{figure}
    \centering
    \includegraphics[scale=1]{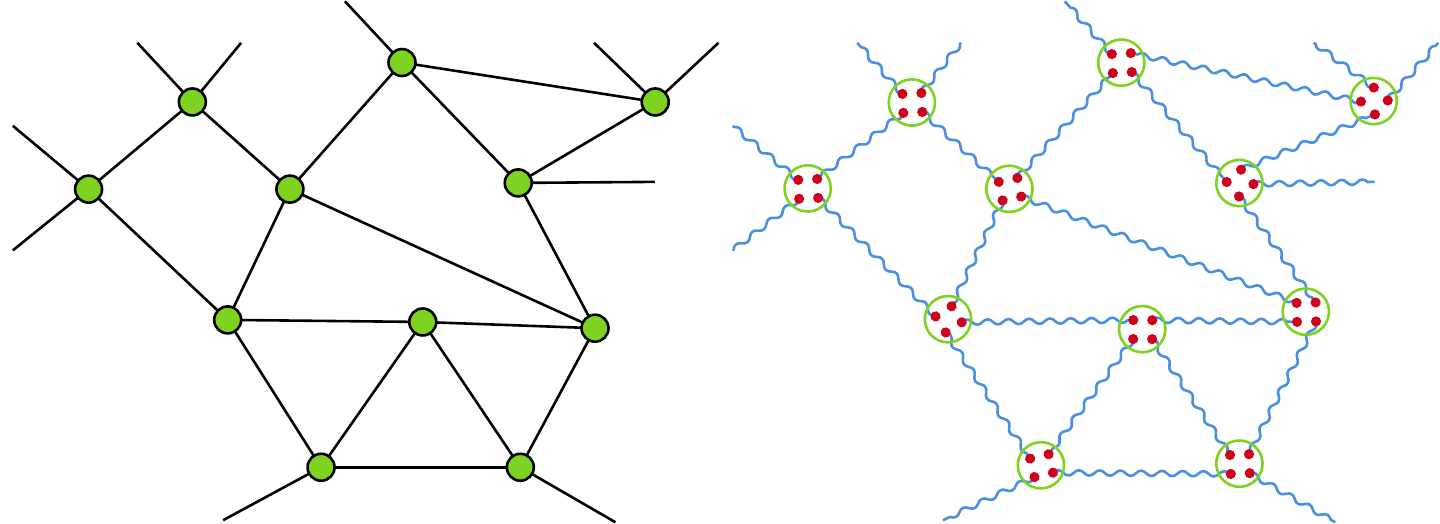}
    \caption{(Left) A degree $d=4$ graph with a Tanner code defined on it. (Right) The associated tensor network, where $d$ qubits (red) are placed on each vertex (green circle) and qubits are connected according to the edge structure using EPR pairs (blue wavy lines). The qubits on each vertex are projected according to the local code.}
    \label{fig:fullgraph}
\end{figure}
Now, we can make the tensor network `injective' by defining 
$$Q_v:=P_v+\delta \id = (1+\delta)P_v+\delta(\id-P_v),\quad \ket{\Psi}:= \frac{1}{\sqrt{Z}}\bigotimes_v Q_v \ket{\Theta_0}.$$
\begin{claim}
\label{clm:psiprop}
It holds that $$Z\leq 2^k(1+2\delta+\delta^22^{(d-k_0)})^n$$ and 
$$|\braket{\Psi}{\Phi}|\geq \frac{(1+\delta)^n}{(1+2\delta+\delta^22^{(d-k_0)})^{\frac{n}{2}}}$$
\end{claim}
\begin{proof}
Consider,
\begin{align*}
    Z&=\bra{\Theta_0}\bigotimes_v Q_v^2 \ket{\Theta_0} = \bra{\Theta_0}\bigotimes_v\br{(1+2\delta)P_v+\delta^2\id}  \ket{\Theta_0}\\
    &=\sum_{S\subset V}(1+2\delta)^{|S|}\delta^{2n-2|S|}\bra{\Theta_0}\bigotimes_{v\in S}P_v \ket{\Theta_0}
\end{align*}
Let us evaluate $\bra{\Theta_0}\bigotimes_{v\in S}P_v \ket{\Theta_0}$. This is essentially the number of codewords when parity checks only act on the vertices in $S$. If we were to include the parity checks in $V\setminus S$ as well, we would obtain the original code. Since there are at most $(d-k_0)(n-|S|)$ independent checks, the following inequality holds:
$$\bra{\Theta_0}\bigotimes_{v\in S}P_v \ket{\Theta_0}\cdot 2^{-(d-k_0)(n-|S|)}\leq 2^k \implies \bra{\Theta_0}\bigotimes_{v\in S}P_v \ket{\Theta_0}\leq 2^k\cdot 2^{(d-k_0)(n-|S|)}.$$
This shows that
\begin{align*}
    Z&\leq 2^k\cdot\sum_{S\subset V}(1+2\delta)^{|S|}\delta^{2n-2|S|}\cdot 2^{(d-k_0)(n-|S|)}=2^k\cdot\sum_{S\subset V}(1+2\delta)^{|S|} \br{\delta^22^{(d-k_0)}}^{(n-|S|)}\\
    &=2^k(1+2\delta+\delta^22^{(d-k_0)})^n.
\end{align*}
Further,
\begin{align*}
    \braket{\Psi}{\Phi}&=\frac{1}{\sqrt{2^kZ}}\bra{\Theta_0}\bigotimes_v Q_vP_v \ket{\Theta_0} =\frac{(1+\delta)^n}{\sqrt{2^kZ}} \bra{\Theta_0}\bigotimes_vP_v  \ket{\Theta_0}\\
    &=\frac{(1+\delta)^n\sqrt{2^k}}{\sqrt{Z}}\geq\frac{(1+\delta)^n}{(1+2\delta+\delta^22^{(d-k_0)})^{\frac{n}{2}}}.
\end{align*}
This completes the proof.
\end{proof}
The nice property of $\ket{\Psi}$ is that it is the unique ground state of a local Hamiltonian. For $e=(v,v')$, define
$$g_e = \br{Q_v\otimes Q_{v'}}^{-1}(\id - \ketbra{\EPR}_{e})\br{Q_v\otimes Q_{v'}}^{-1}, \quad h_e= \text{span}\br{g_e}$$
where `span' means that $h_e$ is the projector onto the image of $g_e$.
Since $g_e\ket{\Psi}=0$, we have $h_e\ket{\Psi}=0$. Let $$H:=\sum_{e\in E} h_e.$$ Then $\ket{\Psi}$ is a ground state of $H$ with ground energy $0$. In fact, we have the following claim, which is well known about injective tensor networks.
\begin{claim}
\label{clm:uniqgs}
$\ket{\Psi}$ is the unique ground state of $H$.
\end{claim}
\begin{proof}
Suppose $\ket{\Psi'}$ is a ground state of $H$. Then it is also a ground state of $\sum_e g_e$. Write $\ket{\Psi'}= \bigotimes_{v\in V} Q_v\ket{\Theta'}$, for a (possibly unnormalized) quantum state $\ket{\Theta'}$. This is possible since $\bigotimes_{v\in V}Q_v$ is invertible. Observe that $\ket{\Theta'}$ is a ground state of $\sum_e (\id - \ketbra{\EPR}_{e})$. This is possible only if $\ket{\Theta'}=\ket{\Theta_0}$, which proves the claim. 
\end{proof}

\section{The Hamiltonian $H$ is cNLTS}
\label{sec:cNLTS}

Suppose $\eps|E|$ local terms from $H$ are removed (see Figure \ref{fig:edgesremoved}). Since each local term corresponds to an edge, let $E_1$ be the remaining edges and let $H_1=\sum_{e\in E_1} h_e$ be the Hamiltonian that remains. We will show that any state $\ket{\psi}$ that is a ground state of $H_1$ has a large circuit complexity, if $\eps$ is a sufficiently small constant.
\begin{figure}
    \centering
    \includegraphics[scale=1]{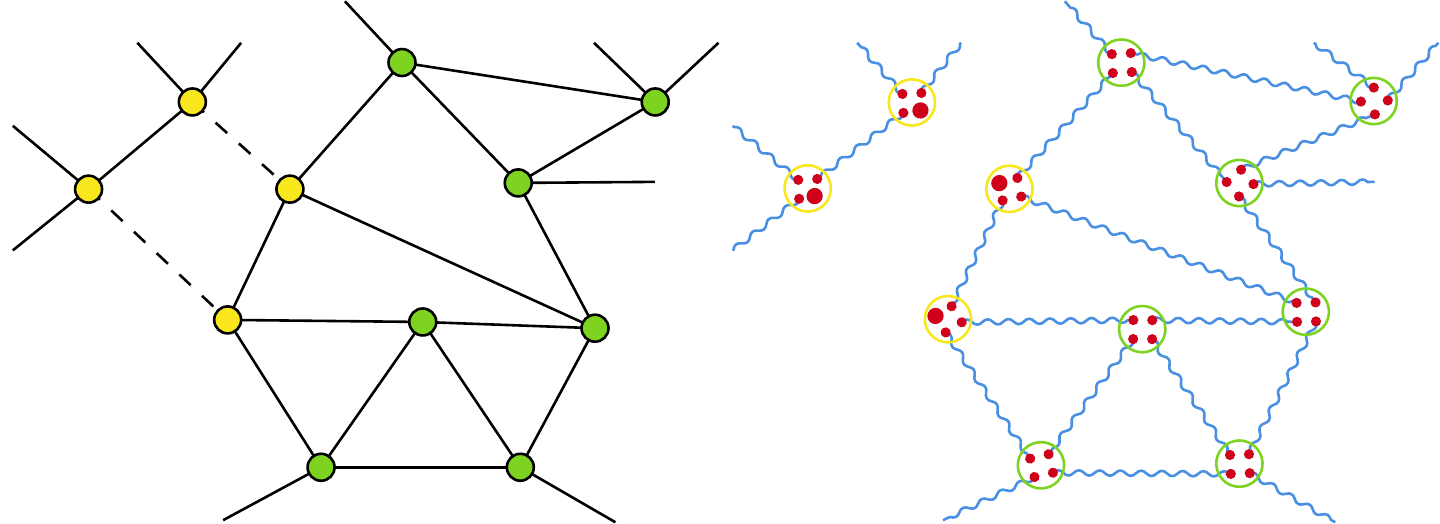}
    \caption{(Left) The dashed edges have been removed. $W$ denotes the set of yellow vertices. The remaining edges are $E_1$. (Right) In the tensor network picture, some qubits are no longer required to be connected by an EPR pair. These qubits, called residuals, are shown as thick red dots inside yellow circles. Their set is $R$.}
    \label{fig:edgesremoved}
\end{figure}

\subsection{Structure of the ground space of $H_1$}
Let $W\subset V$ be the set vertices on which the removed edges were incident. Among the $d|W|$ qubits in these vertices, some qubits were associated to the removed edges. We will call these qubits `residual' and denote their set by $R=\{1,2,\ldots 2\eps|E|\}$. We are free to choose any `assignment' $\ket{0},\ket{1}$ to the residual qubits. Since we have been thinking of assignments as occurring on the edges,  we will sometimes refer to $R$ as a set of edges (Figure~\ref{fig:edgesremoved}). Thus, we will continue using the terminology of `edge assignment'. Note the following cardinality bounds:
\begin{equation}
    \label{eq:G1card}
    |W|\leq |R|= 2\eps|E|, \quad |E_1|=|E|(1-\eps)=\frac{d|V|-|R|}{2}. 
\end{equation}
Let 
$$\ket{\Theta_1}:=\bigotimes_{e\in E_1}\ket{\EPR}_{e}.$$ We have the following claim, which is analogous to Claim \ref{clm:uniqgs}. 
\begin{claim}
\label{clm:gspaceH1}
The ground space of $H_1$ is 
$$\cG:=\text{span}\br{\bigotimes_{v\in V}Q_v(\ket{\Theta_1}\otimes_{r\in R}\ket{b_r}_{r}), \text{ such that }  b:=b_1,\ldots b_{|R|}\in \{0,1\}^{|R|}}.$$
\end{claim}
\begin{proof}
Let $\ket{\omega}=\bigotimes_{v\in V}Q_v\ket{\tau}$ be a ground state of $H_1$ for some $\ket{\tau}$, which is possible since $\bigotimes_{v\in V}Q_v$ is invertible. Note that 
$$H_1\ket{\omega}=0 \implies \forall e\in E_1, \quad h_e\ket{\omega}=0 \implies \forall e\in E_1, \quad g_e\ket{\omega}=0.$$ Thus, for all $e\in E_1$, $\br{\id - \ketbra{\EPR}_e}\ket{\tau}=0$. This shows that $\ket{\tau}$ belongs to the space spanned by the vectors $\{\ket{\Theta_1}\otimes_{r\in R}\ket{b_r}_{r}, \text{ such that }b:=b_1,\ldots b_{|R|}\in \{0,1\}^{|R|}\}$, as there are no constraints on the residual qubits. This completes the proof.
\end{proof}

Consider the following basis within $\cG$: $$\ket{\Psi_b}\propto  \bigotimes_{v\in V}Q_v(\ket{\Theta_1}\otimes_{r\in R}\ket{b_r}_{r}), \forall b\in \{0,1\}^{|R|},$$
where each $\ket{\Psi_b}$ is normalized as $\braket{\Psi_b}{\Psi_b}=1$. Note that this is an orthonormal basis, as the residual qubits are fixed according to $b$ (the operators $Q_v$ do not change any computational basis state). Along the lines of Claim \ref{clm:psiprop}, we would expect that this state is close to the following state (ignoring normalization)
$$\bigotimes_{v\in V}P_v(\ket{\Theta_1}\otimes_{r\in R}\ket{b_r}_{r}), \forall b\in \{0,1\}^{|R|}.$$
But we have to be careful: if $P_v\br{\otimes_{r\in R}\ket{b_r}_{r}}=0$ for any $v\in W$ \footnote{This would happen iff the edge assignment $b$ were simply inconsistent with the codewords allowed at $v$. For example, $v$ may connect with $d-1$ edges in $R$ and $b$ may assign $\Delta_0-2$ number of `$1$'s on those edges.}, the above state is $0$; whereas $\ket{\Psi_b}$ is non-zero for all $b$. With this in mind, we let $W_b\subset W$ denote all the vertices with which $b$ is consistent (Figure \ref{fig:newspace} (left)). 
\begin{figure}
    \centering
    \includegraphics[scale=1]{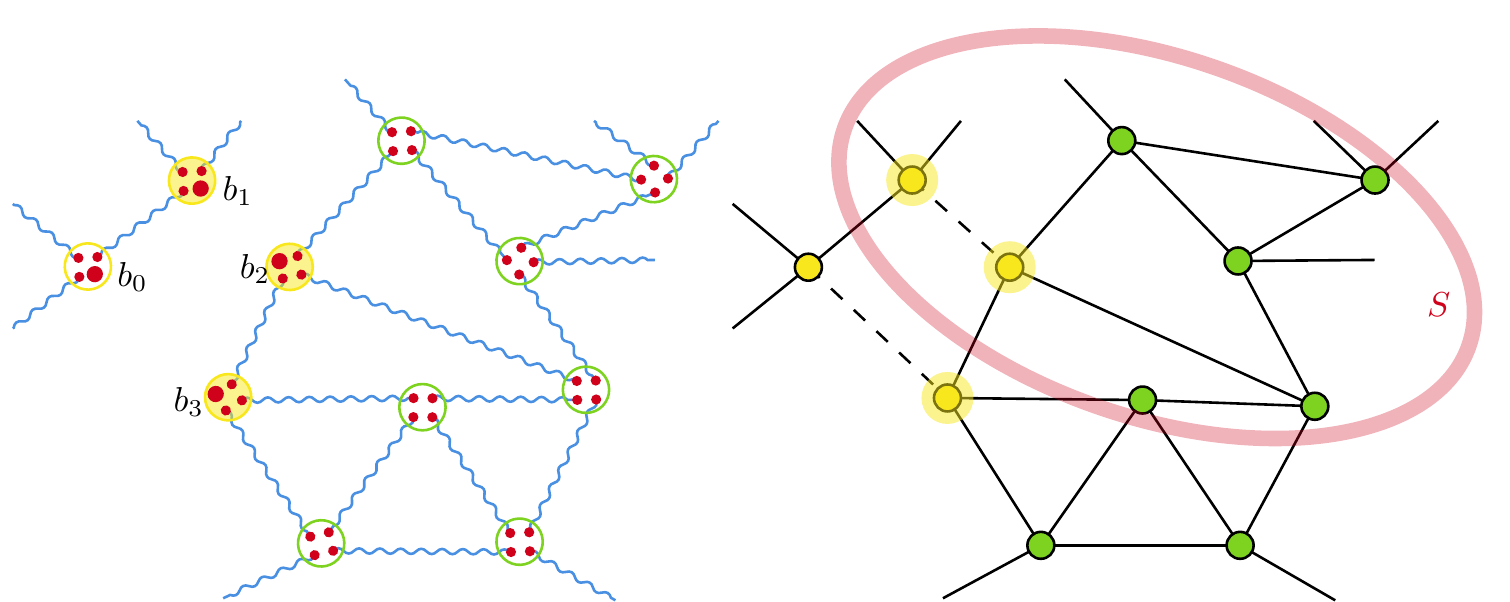}
    \caption{(Left) The residual qubits are no longer connected by EPR pairs. Thus, they can be assigned any computational-basis state $\ket{b}: b\in\{0,1\}^{|R|}$ (in fact, they can be assigned any quantum state on $|R|$ qubits; but we focus on computational-basis states at the moment). A given assignment $b$ may violate checks on some vertices in~$W$ (yellow circle). Here, we depict shaded yellow circles, where $b$ does not cause any violated checks. This is the set $W_b\subset W$. (Right) Vertices in $W_b$ are now depicted by yellow dots with shaded surrounding. In Claim~\ref{clm:psibprop}, a set $S\subset (V\setminus W) \cup W_b$ is considered. Equation~\ref{eq:Sequation} can be verified from here.}
    \label{fig:newspace}
\end{figure}
We observe that
$$\ket{\Psi_b}=\frac{1}{\sqrt{Z_b}}  \bigotimes_{v\in (V\setminus W)\cup W_b}Q_v(\ket{\Theta_1}\otimes_{r\in R}\ket{b_r}_{r}),$$
since $Q_v$ acts as $\delta\id$ at any $v\in W\setminus W_b$ (above $Z_b$ is a normalization constant) and define
$$\ket{\Phi_b}=\frac{1}{\sqrt{2^{k_b}}}  \bigotimes_{v\in (V\setminus W)\cup W_b}P_v(\ket{\Theta_1}\otimes_{r\in R}\ket{b_r}_{r}),$$
where $k_b$ will be determined shortly. Note that the states $\{\ket{\Phi_b}\}_{b\in \{0,1\}^{|R|}}$ are mutually orthogonal. The following claim is analogous to Claim \ref{clm:psiprop}.
\begin{claim}
\label{clm:psibprop}
It holds that
$$k_b \geq \br{\frac{2k_0}{d}-1}|E|-|R|$$
and
$$|\braket{\Phi_b}{\Psi_b}|\geq \frac{(1+\delta)^n}{(1+2\delta+\delta^22^{(d-k_0)})^{\frac{n}{2}}}.$$
\end{claim}
\begin{proof}
We write down an expression for $k_b$. Note that $\ket{\Phi_b}$ is simply a superposition over edge assignments that satisfy the Tanner code with checks on $(V\setminus W)\cup W_b$, where we condition the edges in $R$ to have fixed edge assignments according to $b$. Conditioning the edge assignments in $R$ to be $b$ leads to a set of parity check over edges in $E_1$ (some of these checks may also impose a parity of $1$ on the edge assignments in $E_1$). Each vertex in $V\setminus W$ contributes to $d-k_0$ checks. Each vertex in $W_b$ contributes to anywhere between $0$ to $d-k_0$ independent checks. If we define $c_v$ as the number of independent checks due to $v\in W_b$ that involve edges in $E_1$, we have 
\begin{equation}
    \label{eq:kbval}
k_b \geq |E_1|-(d-k_0)(|V\setminus W|)-\sum_{v\in W_b} c_v.
\end{equation}
Since $c_v\leq d-k_0$ and $W_b\subseteq W$, a lower bound on $k_b$ is
\begin{align*}
k_b &\geq |E_1|-(d-k_0)(|V\setminus W|+|W_b|)\geq |E_1|- (d-k_0)|V|\\
& \overset{Eq. \ref{eq:G1card}}{=} |E_1|- \frac{d-k_0}{d}(2|E_1|+|R|) = \br{\frac{2k_0}{d}-1}|E_1|-\frac{d-k_0}{d}|R|\\
&\geq \br{\frac{2k_0}{d}-1}|E_1|-|R|.
\end{align*}
Next,
\begin{align*}
    Z_b&=\bra{\Theta_1}\otimes_{r\in R}\bra{b_r}_r  \bigotimes_{v\in (V\setminus W)\cup W_b}Q^2_v\ket{\Theta_1}\otimes_{r\in R}\ket{b_r}_r   \\ 
    &= \bra{\Theta_1}\otimes_{r\in R}\bra{b_r}_r  \bigotimes_{v\in (V\setminus W)\cup W_b}((1+2\delta)P_v+\delta^2\id)\ket{\Theta_1}\otimes_{r\in R}\ket{b_r}_r\\
    &= \sum_{S\subset (V\setminus W)\cup W_b}(1+2\delta)^{|S|}\delta^{2(|(V\setminus W)\cup W_b|-|S|)}\bra{\Theta_1}\otimes_{r\in R}\bra{b_r}_r  \bigotimes_{v\in S}P_v\ket{\Theta_1}\otimes_{r\in R}\ket{b_r}_r\\
    &\overset{(1)}{\leq} 2^{k_b}\sum_{S\subset (V\setminus W)\cup W_b}(1+2\delta)^{|S|}\delta^{2(|(V\setminus W)\cup W_b|-|S|)}\cdot 2^{(d-k_0)(|V\setminus W|-|S\setminus W|)+\sum_{v\in W_b\setminus S}c_v}\\
    &\leq 2^{k_b}\sum_{S\subset (V\setminus W)\cup W_b}(1+2\delta)^{|S|}\delta^{2(|(V\setminus W)\cup W_b|-|S|)}\cdot 2^{(d-k_0)(|V\setminus W|-|S\setminus W|+|W_b\setminus S|)}\\
    &\overset{(2)}{=}2^{k_b}\sum_{S\subset (V\setminus W)\cup W_b}(1+2\delta)^{|S|}\br{\delta^2\cdot 2^{d-k_0}}^{(|(V\setminus W)\cup W_b|-|S|)}\\
    &=2^{k_b}(1+2\delta+\delta^2\cdot 2^{d-k_0})^{(|(V\setminus W)\cup W_b|)}.
\end{align*}
For $(1)$, note that $\bra{\Theta_1}\otimes_{r\in R}\bra{b_r}_r  \bigotimes_{v\in S}P_v\ket{\Theta_1}\otimes_{r\in R}\ket{b_r}_r$ is the number of edge assignments that satisfy parity checks in $S$ (with the condition that edges in $R$ are assigned $b$). If we were to add parity checks on remaining vertices in $(V\setminus W) \setminus S$ and $W_b \setminus S$, we would obtain the $2^{k_b}$ codewords accounted for in Eq. \ref{eq:kbval}. The number of such linearly independent parity checks that are added is at most $(d-k_0)(|V\setminus W|-|S\setminus W|)+\sum_{v\in W_b\setminus S} c_v$. This gives us the upper bound
$$\bra{\Theta_1}\otimes_{r\in R}\bra{b_r}_r  \bigotimes_{v\in S}P_v\ket{\Theta_1}\otimes_{r\in R}\ket{b_r}_r\cdot 2^{-(d-k_0)(|V\setminus W|-|S\setminus W|)-\sum_{v\in W_b\setminus S} c_v}\leq 2^{k_b},$$
and $(2)$ uses (see Figure \ref{fig:newspace}(right)) 
\begin{equation}
    \label{eq:Sequation}
    |V\setminus W|-|S\setminus W|+|W_b\setminus S|=|V\setminus W|+|W_b|-|S|=|(V\setminus W)\cup W_b|-|S|,
\end{equation}
where we repeatedly used the fact that $S\subset (V\setminus W)\cup W_b$. Thus,
\begin{align*}
    |\braket{\Phi_b}{\Psi_b}|&=\frac{1}{\sqrt{2^{k_b}Z_b}}\bra{\Theta_1}\otimes_{r\in R}\bra{b_r}_r  \bigotimes_{v\in (V\setminus W)\cup W_b}P_vQ_v\ket{\Theta_1}\otimes_{r\in R}\ket{b_r}_r   \\ 
    &= \frac{(1+\delta)^{|(V\setminus W)\cup W_b|}}{\sqrt{2^{k_b}Z_b}}\bra{\Theta_1}\otimes_{r\in R}\bra{b_r}_r  \bigotimes_{v\in (V\setminus W)\cup W_b}P_v\ket{\Theta_1}\otimes_{r\in R}\ket{b_r}_r\\
    &=\frac{(1+\delta)^{|(V\setminus W)\cup W_b|}}{\sqrt{2^{k_b}Z_b}}\cdot 2^{k_b}= (1+\delta)^{|(V\setminus W)\cup W_b|}\cdot\sqrt{\frac{2^{k_b}}{Z_b}}\\
    &\geq (1+\delta)^{|(V\setminus W)\cup W_b|}\cdot\frac{1}{\sqrt{(1+2\delta+\delta^2\cdot 2^{d-k_0})^{(|(V\setminus W)\cup W_b|)}}}\\
    &\geq \frac{(1+\delta)^n}{(1+2\delta+\delta^22^{(d-k_0)})^{\frac{n}{2}}}.
\end{align*} 
Above, last inequality holds since $|(V\setminus W)\cup W_b|\leq n$.
\end{proof}
Thus, if $\delta$ is small enough, it suffices to understand the properties of the space 
$$\cG':=\text{span}\br{\ket{\Phi_b}, \text{ such that } b\in \{0,1\}^{|R|}}.$$
From Claim \ref{clm:psibprop}, each $\ket{\Phi_b}$ is a superposition over $2^{k_b}$ edge assignments. Moreover, each such edge assignment satisfies all the checks in $V':= V\setminus W$. Thus, let us understand the properties of edge assignments that satisfy such checks. Define a new Tanner code $T':=T(C, G')$, where $G'=(V', E'\cup F)$, $E'$ is the set of edges in the subgraph induced by~$V'$ and $F$ is the set of edges which connected $V'$ with $W$. We will think of each edge in $F$ as `free', being incident to just one vertex in $V'$ (Figure \ref{fig:Tprime}). 
\begin{figure}
    \centering
    \includegraphics[scale=1]{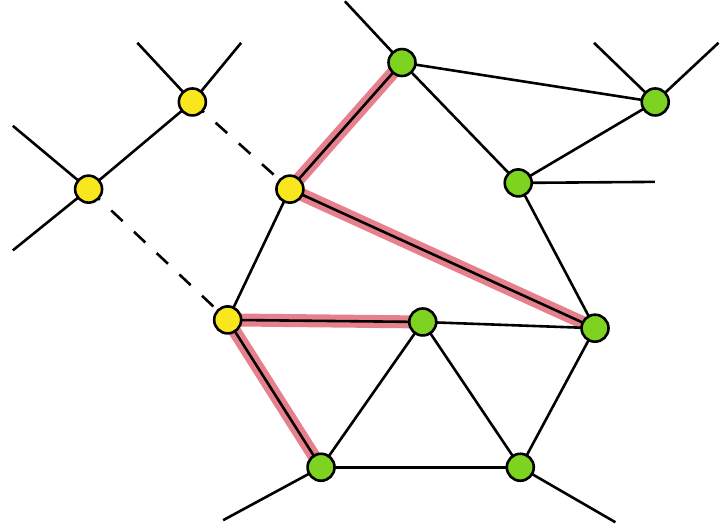}
    \caption{Any edge assignment appearing in $\ket{\Phi_b}$ still satisfies all the checks at the green vertices ($V'$). If we restrict to $V'$, the shaded edges have only one end-point in $V'$. We will call such edges $F$. $I\subset V'$ is the set of vertices on which the edges in $F$ are incident. In Theorem~\ref{thm:dualdist}, some care is needed in analyzing the Hamming weight of edge assignments satisfying checks on $V'$, as vertices in $I$ are responsible for the breakdown of expansion in $V'$. While checks on $I$ are satisfied, this may happen due to edges in $F$ and may not contribute to expansion.}
    \label{fig:Tprime}
\end{figure}
Theorem \ref{thm:dualdist} below shows that the Hamming distance between the codewords of $T'$ is either small or large.

\subsection{Properties of the Tanner code $T'$}
\label{sec:Tprimeprop}
Let $I\subset V'$ be the set of vertices on which an edge in $F$ is incident (Figure \ref{fig:Tprime}). From Eq. \ref{eq:G1card}, note the following bounds 
\begin{equation}
\label{eq:connectedbounds}
|I|\leq |F|\leq d|W|\leq 2d\eps|E| \text{ and }|E|\geq |E'| \geq |E_1| - d|W|\geq|E|\br{1-3d\eps}.    
\end{equation}
Since $G'$ may no longer be an expander, we do not have any guarantees on the distance of~$T'$. But we can show some structure in the codewords of $T'$.
\begin{theorem}
\label{thm:dualdist}
Suppose $\eps\leq\frac{1}{6d}$ and $\Delta_0\geq 4\lambda$. The Hamming distance between the codewords of $T'$ is either $\leq 8d^2\eps|E'\cup F|$ or $\geq \frac{\Delta_0^2}{8d^2}\cdot |E'\cup F|$.   
\end{theorem}
\begin{proof}
Since $T'$ is a linear code, we show that the Hamming weight of a non-zero codeword~$x$ is either $\leq 8d^2\eps|E'\cup F|$ or $\geq \frac{\Delta_0^2}{8d^2}\cdot |E'\cup F|$. Let $J_x\subset E'\cup F$ be the set of edges on which $x$ assigns $1$. Suppose $|J_x|\leq 8d^2\eps|E'\cup F|$; then we are done. Else 
$$|J_x|\geq 8d^2\eps|E'\cup F| \geq 8d^2\eps|E'|\overset{Eq. \ref{eq:connectedbounds}}{\geq} 8d^2\eps|E|\br{1-3d\eps}\geq 4d^2\eps|E|.$$ We consider this case and show that $|J_x|$ must be significantly larger. Let $S\subset V'$ be the vertices on which the edges in $J_x$ are incident. We can apply expander mixing lemma to the \textit{original graph} $G$ and obtain $$|E(S,S)|\leq \frac{d|S|^2}{n} + \lambda|S|.$$
On the other hand, we can lower bound $|E(S,S)|$ as follows. Every vertex in $S\setminus I$ has degree at least $\Delta_0$ and each edge incident to such a vertex belongs to $E(S,S)$ (since such an edge is not in $F$, both its endpoint are in $S$). Edges incident to vertices in $S\cap I$ may not belong to $E(S,S)$, thus we won't count them (see Figure \ref{fig:Tprime}). Overall, we have 
\begin{equation}
\label{eq:SSedges}
|E(S,S)|\geq \Delta_0|S\setminus I| \overset{Eq. \ref{eq:connectedbounds}}{\geq} \Delta_0(|S|-2d\eps|E|)=\Delta_0\br{\frac{|S|}{2}+\frac{|S|}{2}-2d\eps|E|}.    
\end{equation}
Since $|J_x|\geq 4d^2\eps|E|$, we can naively bound $|S|\geq \frac{|J_x|}{d}\geq 4d\eps|E|$. Thus,
$$\frac{d|S|^2}{n} + \lambda|S|\geq |E(S,S)| \overset{Eq. \ref{eq:SSedges}}{\geq} \frac{\Delta_0}{2}|S|\overset{\lambda\leq \frac{\Delta}{4}}{\implies} |S|\geq \frac{n\Delta_0}{4d}.$$ From here, we obtain 
$$|J_x|\overset{(1)}{\geq} \frac{\Delta_0|S|}{2} \geq \frac{n\Delta_0^2}{16d} = |E|\cdot \frac{\Delta_0^2}{8d^2} \geq |E'\cup F|\cdot \frac{\Delta_0^2}{8d^2}. $$
Here, $(1)$ follows since every vertex in $S$ is associated to at least $\Delta_0$ edges in $J_x$ and every edge in $J_x$ is associated to at most $2$ vertices in $S$. This completes the proof.
\end{proof}

\subsection{Structure of the states $\ket{\Phi_b}$}

Recall that $\ket{\Phi_b}$ is a superposition over the edge assignments to $E_1\cup R$. We now show that these edge assignments form distant clusters. 
\begin{theorem}
\label{thm:phiprop}
Suppose $\eps\leq \frac{\Delta_0^2}{300d^4}$ and $\Delta_0\geq 4\lambda$. There are disjoint sets $B_1, B_2, \ldots \subset \{0,1\}^{|E_1\cup R|}$ such that for any $x,y\in B_i$, the Hamming distance between $x$ and $y$ is $\leq 10d^2\eps|E_1\cup R|$ and for any $x\in B_i$ and $y\in B_j$ with $i\neq j$, the Hamming distance between $x,y$ is $\geq\frac{\Delta_0^2}{10d^2}\cdot |E_1\cup R|$. Furthermore, the states $\{\ket{\Phi_b}\}_{b\in \{0,1\}^{|R|}}$ are uniform superpositions over some edge assignments in $\cup_i B_i$.
\end{theorem}
\begin{proof}
Any two edge assignments $x,y$ appearing in $\ket{\Phi_b}$, when restricted to the edges in $G'$ (which we denote $x_{G'}, y_{G'}$), belong to $T'$. The edge assignment $x$ is obtained from $x_{G'}$ by appending assignments to the edges in $(E_1\cup R) \setminus (E'\cup F)$. There are at most $d|W| \overset{Eq. \ref{eq:G1card}}{\leq} 2d\eps |E|$ such edges. Thus, invoking Theorem \ref{thm:dualdist}, the Hamming distance between $x,y$ is either at most 
$$8d^2\eps|E'\cup F|+2d\eps|E|\leq 8d^2\eps|E_1\cup R|+2d\eps|E| \overset{Eq \ref{eq:G1card}}{\leq} 8d^2\eps|E_1\cup R|+2d\eps|E_1\cup R|\leq 10d^2\eps|E_1\cup R|$$
or at least
\begin{align*}
&\frac{\Delta_0^2}{8d^2}\cdot |E'\cup F|-2d\eps|E| \geq \frac{\Delta_0^2}{8d^2}\cdot |E_1\cup R|-\frac{\Delta_0^2}{8d^2}\cdot 2d\eps|E|- 2d\eps|E|\\
&\overset{Eq. \ref{eq:G1card}}{\geq} \frac{\Delta_0^2}{8d^2}\cdot |E_1\cup R| - 4d\eps |E_1\cup R|\geq \frac{\Delta_0^2}{10d^2}\cdot |E_1\cup R|.    
\end{align*}
Next, consider a relation $\cR$ between the edge assignments: $x,y\in \cR$ if the Hamming distance between them is $\leq 10d^2\eps|E_1\cup R|$. The relation is transitive: $x,y\in \cR$ and $y,z\in \cR$ implies the Hamming distance between $x,z$ is $\leq 20d^2\eps|E_1\cup R|< \frac{\Delta_0^2}{10d^2}\cdot |E_1\cup R|$, which in turn requires that the Hamming distance between $x,z$ is $\leq 10d^2\eps|E_1\cup R|$. This forces $x,z\in \cR$. The sets $B_1, B_2, \ldots$ are the equivalence classes formed by this relation, which completes the proof.
\end{proof}
Let $$\Pi_{B_i}:= \sum_{x\in B_i}\bigotimes_{e\in E_1}\ketbra{x_e}_e \bigotimes_{r\in R} \ketbra{x_r}_r$$ be the projector onto the edge assignments in $B_i$. The following claim holds.
\begin{claim}
\label{clm:noconc}
Let $\delta^2 \leq \frac{2^{-d}}{10000n}$, $\eps \leq \frac{1}{20000d^2}$  and $k_0\geq 0.55d$. For any $i$ and any $b, b'\in \{0,1\}^{|R|}$ (with $b\neq b'$), $\bra{\Psi_b}\Pi_{B_i}\ket{\Psi_b}\leq \frac{1}{50}$ and $\bra{\Psi_b}\Pi_{B_i}\ket{\Psi_{b'}}=0$.
\end{claim}
\begin{proof}
Since $\ket{\Phi_b}$ is a uniform superposition over $2^{k_b}$ edge assignments and the size of each $B_i$ is at most ${|E_1\cup R| \choose 10d^2\eps |E_1\cup R|}\leq 2^{2d\sqrt{10\eps}|E_1\cup R|} \leq 2^{8d\sqrt{\eps}|E|}$, we have 
$$\bra{\Phi_b}\Pi_{B_i}\ket{\Phi_b}\leq 2^{-k_b}\cdot 2^{8d\sqrt{\eps}|E|} \overset{Claim \ref{clm:psibprop}}{\leq} 2^{8d\sqrt{\eps}|E| - 0.1|E| +|R|}\overset{Eq. \ref{eq:G1card}}{\leq} 2^{10d\sqrt{\eps}|E|-0.1|E|}\leq \frac{1}{100}.$$
Above, the last inequality assumes that $|E|$ is larger than some constant. Now, Claim \ref{clm:psibprop} ensures that 
$$|\braket{\Phi_b}{\Psi_b}|\geq \frac{(1+\delta)^n}{(1+2\delta+\delta^22^{(d-k_0)})^{\frac{n}{2}}}\geq \br{\frac{1}{1+2^d\delta^2}}^{\frac{n}{2}}\geq e^{-\frac{1}{20000}}\geq 1-\frac{1}{10000}.$$
Thus, $\frac{1}{2}\|\ketbra{\Phi_b}-\ketbra{\Psi_b}\|_1\leq \frac{1}{100}$, which ensures that 
$$\Tr{\Pi_{B_i}\ketbra{\Psi_b}}\leq \Tr{\Pi_{B_i}\ketbra{\Phi_b}}+\frac{1}{100}\leq \frac{1}{50}.$$
To argue that $\bra{\Psi_b}\Pi_{B_i}\ket{\Psi_{b'}}=0$, note that $\ket{\Psi_{b}}$ is a superposition over edge assignments with the fixed $b$ on the edges in $R$. That is, $\ket{\Psi_{b}} = \ket{\ldots} \otimes_{r\in R} \ket{b_r}_r$. Thus, $\Pi_{B_i}$, being a projector onto computational basis states, satisfies $\Pi_{B_i}\ket{\Psi_{b}} = \ket{\ldots}\otimes_{r\in R} \ket{b_r}_r$. Similarly, $\ket{\Psi_{b'}} = \ket{\ldots'} \otimes_{r\in R} \ket{b'_r}_r$. Since $b\neq b'$, the claim follows.  
\end{proof}
\subsection{Circuit lower bound}
We will assume that $\eps\leq \frac{1}{20000d^2}, \frac{\Delta_0^2}{300d^4}$; $\Delta_0\geq 4\lambda$; $k_0\geq 0.55d$ and $\delta^2\leq \frac{2^{-d}}{10000n}$. Note that these conditions can be met with constant $k_0,\Delta_0, d$, which ensures that $\eps$ is a constant (see Section~\ref{sec:explicit}). Our main theorem is below, which proves that $H$ is cNLTS. The argument is directly inspired by the quantum circuit lower bound argument in \cite{EldarH17}, based on the partitioning of quantum codewords. However, we consider a simpler argument based on the tight polynomial approximations to the AND function \cite{KahnLS96, BuhrmanCWZ99, AAG21}, inspired by \cite{KAAV15}.
\begin{theorem}
\label{thm:mainlb}
Let $\ket{\rho}=U\ket{0}^{\otimes m}$ on $m\geq nd$ qubits, where $U$ is a depth $t$ quantum circuit, such that $\frac{1}{2}\|\ketbra{\Gamma}-\ketbra{\rho}\|_1\leq 0.1$ for some ground state \footnote{While $\ket{\Gamma}$ may act on more qubits than included in $V$, we say that $\ket{\Gamma}$ is a ground state of $H_1$ if its reduced state on the qubits in $V$ belongs to the ground space of $H_1$.} $\ket{\Gamma}$ of $H_1$. It holds that 
$$t= \Omega\br{\log \frac{n\Delta_0^4}{d^3}}.$$
\end{theorem}
\begin{proof}
Note that $m\leq 2^tnd$ without loss of generality, as $H_1$ acts on $nd$ qubits (see \cite[Section 2.3]{ANirkhe22} for a justification based on the light cone argument). We can expand 
$$\ket{\Gamma}= \sum_{b\in \{0,1\}^{|R|}} \ket{\mu_b}\otimes \ket{\Psi_b},$$
such that $\sum_{b\in \{0,1\}^{|R|}} \|\ket{\mu_b}\|^2=1$. The (possibly unnormalized) vectors $\ket{\mu_b}$ act on $m-nd$ qubits outside $V$. Using Claim \ref{clm:noconc}, we find that for any $i$, \begin{equation}
    \label{eq:noconc}
    \bra{\Gamma}\Pi_{B_i}\ket{\Gamma}= \sum_{b,b'\in \{0,1\}^{|R|}} \braket{\mu_{b'}}{\mu_b} \bra{\Psi_{b'}}\Pi_{B_i}\ket{\Psi_b}=\sum_{b\in \{0,1\}^{|R|}} \|\ket{\mu_b}\|^2 \bra{\Psi_{b}}\Pi_{B_i}\ket{\Psi_b}\leq \frac{1}{50}.
\end{equation}
On the other hand, all edge assignments over $E_1\cup R$ appearing in $\ket{\Gamma}$ belong to some $B_i$. In other words, $\sum_i \bra{\Gamma}\Pi_{B_i}\ket{\Gamma}=1$. Thus, we can find two disjoint sets of indices $M,M'$, such that $$\sum_{i\in M} \bra{\Gamma}\Pi_{B_i}\ket{\Gamma}\geq \frac{1}{2}-\frac{1}{50}\geq \frac{1}{3},\quad \sum_{i\in M'} \bra{\Gamma}\Pi_{B_i}\ket{\Gamma}\geq \frac{1}{2}-\frac{1}{50}\geq \frac{1}{3}.$$ Define $B_M=\cup_{i\in M} B_i$, $B_{M'}=\cup_{i\in M'} B_i$, $\Pi_M=\sum_{i\in M}\Pi_{B_i}$ and $\Pi_{M'}=\sum_{i\in M'}\Pi_{B_i}$. From Theorem \ref{thm:phiprop}, the Hamming distance between the sets $B_M$ and $B_{M'}$ is $\geq \frac{\Delta_0^2}{10d^2}\cdot|E_1\cup R| \geq \frac{\Delta_0^2}{20d^2}\cdot nd$. On the other hand, we just established that $\bra{\Gamma}\Pi_M\ket{\Gamma}\geq \frac{1}{3}$ and $\bra{\Gamma}\Pi_{M'}\ket{\Gamma}\geq \frac{1}{3}$. From \cite{KahnLS96, BuhrmanCWZ99, AAG21}, there exists a $f\cdot 2^t$-local operator $L$ such that 
$$\|\ketbra{\rho}-L\|_\infty\leq e^{-\frac{f^2}{2^t\cdot 100nd}}.$$
Setting $f\cdot 2^t=\frac{\Delta_0^2}{100d^2}\cdot nd$, we obtain
$$\|\ketbra{\rho}-L\|_\infty\leq e^{-n\cdot\frac{\Delta_0^4}{2^{3t}\cdot 10^6d^3}}.$$
Since $\Pi_M L \Pi_{M'}=0$, we have 
$$\|\Pi_M\ketbra{\rho}\Pi_{M'}\|_\infty \leq e^{-n\cdot\frac{\Delta_0^4}{2^{3t}\cdot 10^6d^3}}.$$
However 
\begin{align*}
&\|\Pi_M\ketbra{\rho}\Pi_{M'}\|_\infty=\sqrt{\bra{\rho}\Pi_M\ket{\rho}\cdot \bra{\rho}\Pi_M\ket{\rho}}\\
&\geq \sqrt{\br{\bra{\Gamma}\Pi_M\ket{\Gamma}-0.1}\cdot \br{\bra{\Gamma}\Pi_M\ket{\Gamma}-0.1}} \geq e^{-2}.    
\end{align*}
Thus, 
$2^{3t}\geq n\cdot \frac{\Delta_0^4}{2\cdot 10^6\cdot d^3}$, which completes the proof.
\end{proof}

\section{Explicit construction of the Tanner code $T(C,G)$}
\label{sec:explicit}

In this section we give an explicit construction of a suitable family of Tanner codes~$\lbrace T(C,G_i) \rbrace_{i=1}^\infty$ from which the family of cNLTS-Hamiltonians $\{H^{(i)}\}_{i=1}^\infty$ of \Cref{thm:cNLTS} is obtained.

For the graphs $G_i$ underlying the Tanner codes, we employ a construction of spectral Cayley-expanders due to Lubotzky, Phillips and Sarnak.
\begin{theorem}[\cite{lubotzky1988ramanujan}]\label{thm:lps}
Assume that~$p$ and~$q$ are distinct, odd primes such that $q > 2 \sqrt{p}$ and $q$ is a square modulo $p$.
Then there exists a symmetric generating set~$\Gamma$ of $\operatorname{PSL}_2(\mathbb{F}_q)$ such that the Cayley graph $\operatorname{Cay}(\operatorname{PSL}_2(\mathbb{F}_q),\Gamma)$ is a non-bipartite, $p+1$-regular expander graph with $\lambda < 2\sqrt{p}$.
\end{theorem}
By fixing a suitable prime $p$ we obtain a family of regular graphs $G_i$ of degree $d=p+1$ and order $q(q^2-1)/2$ with spectral bound $\lambda < 2\sqrt{d-1}$.
We mention in passing that by a result due to Alon and Boppana, this is the best possible bound that any family of regular graphs can achieve and such families are called \emph{Ramanujan graphs}.

Further, we require a linear, binary code $C$.
More specifically, for the construction of the cNLTS-Hamiltonians to go through, we require the existence of a classical linear binary code~$C$ of block size~$d$ encoding at least $k_0 \geq 0.55\, d$ bits with distance $\Delta_0 \geq 4 \lambda$.
As the degree~$d$ of the graphs is constant, it suffices to show that a suitable code $C$ exist, as a brute-force search has time-complexity bounded by a constant $O(1)$.
The existence of a suitable code is guaranteed by the Gilbert--Varshamov bound (see e.g. \cite[Chapter~5]{van2012introduction} for a proof).
\begin{theorem}[Gilbert-Varshamov]\label{thm:gv}
Let $0\leq \mu \leq 0.5$ and let $0\leq R_0\leq 1-H_2(\mu)$,
then there exists a binary, linear code of block size $d$, rank $k_0=R_0d$ and distance $\Delta_0 = \mu d$.
\end{theorem}

We can now give an explicit construction of a suitable code family $\lbrace T(C,G_i) \rbrace_{i=1}^\infty$.
By Theorem~\ref{thm:gv} there exists a code $C$ encoding $k_0 \geq 0.55 d$ bits when $\mu \leq 0.09$.
We further require that $C$ has distance $\Delta_0 = \mu d \geq 4\lambda$.
This is the case via Theorem~\ref{thm:lps} by choosing, for example, $\mu = 0.09$ and prime $p = 7901$.

\subsection*{Acknowledgement}
We thank Chinmay Nirkhe for helpful discussions and Robbie King for carefully reading our manuscript. AA acknowledges support through the NSF award QCIS-FF: Quantum Computing \& Information Science Faculty Fellow at Harvard University (NSF 2013303). NPB acknowledges support through the EPSRC Prosperity Partnership in Quantum Software for Simulation and Modelling (EP/S005021/1).

\bibliographystyle{aipauth4-1}
\bibliography{references}
\end{document}